\documentclass[11pt]{article}

\usepackage{appendix}
\usepackage{amsmath}
\usepackage{amssymb}
\usepackage{amsthm}
\usepackage{multirow}
\usepackage{color}
\usepackage{longtable}
\usepackage{array}
\usepackage{url}
\usepackage{booktabs}
\usepackage{float}
\usepackage{mathtools}
\usepackage{tikz}

\newtheorem{theorem}{Theorem}[section]
\newtheorem{definition}[theorem]{Definition}
\newtheorem{lemma}[theorem]{Lemma}
\newtheorem{example}[theorem]{Example}

\newtheorem{corollary}[theorem]{Corollary}

\newtheorem{remark}[theorem]{Remark}

\title{New Theoretical Bounds and Constructions of \\ Permutation Codes under Block Permutation Metric}

\author{Zixiang Xu$^{\text{a,}}$\thanks{Email address: zxxu8023@qq.com.},~~Yiwei Zhang$^{\text{b,}}$\thanks{Supported in part by a Technion fellowship. Email address: ywzhang@cs.technion.ac.il.}~~and Gennian Ge$^{\text{a,}}$\thanks{Corresponding author. Email address: gnge@zju.edu.cn. Research supported by the National Natural Science Foundation of China under Grant Nos. 11431003 and 61571310, Beijing Scholars Program, Beijing Hundreds of Leading Talents Training Project of Science and Technology, and Beijing Municipal Natural Science Foundation.}\\
\footnotesize $^{\text{a}}$ School of Mathematical Sciences, Capital Normal University, Beijing, 100048, China\\
\footnotesize $^{\text{b}}$ Department of Computer Science, Technion -- Israel Institute of Technology, Haifa 32000, Israel.\\}

\begin{document}

\date{}

\maketitle

\begin{abstract}
Permutation codes under different metrics have been extensively studied due to their potentials in various applications. Generalized Cayley metric is
introduced to correct generalized transposition errors, including previously studied metrics such as Kendall's $\tau$-metric, Ulam metric and Cayley
metric as special cases. Since the generalized Cayley distance between two permutations is not easily computable, Yang et al. introduced a related metric of the same order,
named the block permutation metric. Given positive integers $n$ and $d$, let $\mathcal{C}_{B}(n,d)$ denote the maximum size of a permutation code in $S_n$ with
minimum block permutation distance $d$. In this paper, we focus on the theoretical bounds of $\mathcal{C}_{B}(n,d)$ and the constructions of  permutation codes under block permutation metric.
Using a graph theoretic approach, we improve the Gilbert-Varshamov type bound by a factor of $\Omega(\log{n})$, when $d$ is fixed and $n$ goes into infinity.
We also propose a new encoding scheme based on binary constant weight codes. Moreover, an upper bound beating the sphere-packing type bound is given when $d$
is relatively close to $n$.
\medskip


\smallskip


\smallskip
\end{abstract}

\section{Introduction}

Let $S_n$ be the symmetric group on $n$ elements. A permutation code is a subset of $S_n$ with some certain constraints. Permutation codes under several
different metrics are widely used due to their various applications. Especially in recent years, permutation codes under Kendall's $\tau$-metric, Ulam
metric and Cayley metric have been extensively studied in clouding storage systems, genome resequencing and the rank modulation scheme of flash memories
\cite{2015BuzagloKendall},\cite{2016BuzagloMulti},\cite{2013FarnoudUlam},\cite{Gologlu2015New},\cite{2014Multipermutation},\cite{1990KendallRankMethod},\cite{2016WangPerms},\cite{2016ZhangSnake}.
Under these metrics, codes are designed to correct transposition errors or translocation errors. In \cite{2014CheeBreakpoint}, Chee and Vu introduced the
generalized Cayley metric which includes the metrics aforementioned as special cases. However, the generalized Cayley distance between two permutations
is in general not easily computable and thus the construction of codes is difficult. In \cite{Yang2018Theoretical}, Yang et al. introduced the block
permutation metric which could be simply computed and is of the same order as the generalized Cayley metric. By the metric embedding method, the problem
of constructing codes in the generalized Cayley metric is transformed into constructing codes in the block permutation metric. Several theoretical
bounds (Gilbert-Varshamov type and sphere-packing type) and constructions of codes under block permutation metric are shown in \cite{Yang2018Theoretical}.

In this paper we further consider permutation codes in $S_n$ under the block permutation metric. We first establish a connection between permutation
codes and independent sets in a corresponding graph and then study the bounds of the independence number of the graph. By this graph theoretic approach,
we improve the Gilbert-Varshamov type bound asymptotically by a factor of $\Omega(\log{n})$, when the minimum distance $d$ is fixed while $n$ goes into
infinity. We also propose a new encoding scheme based on certain constructions of binary constant weight codes. Compared with the known constructions,
we improve the size of codes by a factor of $\Theta(n^{2d-4})$. As for the upper bound, each permutation can be represented as a corresponding characteristic set
and then we apply some methods from extremal set theory to obtain an upper bound of a new type, which beats the sphere-packing type bound when $d$ is relatively close to $n$.

The rest of this paper is organized as follows. In Section \ref{sec:pre}, we review some basic backgrounds about block permutation metric. In Section
\ref{sec:graph}, we introduce some relevant terminologies and results from extremal graph theory and then establish the correspondence between
permutation codes and independent sets in some certain graph. The asymptotic improvement of the Gilbert-Varshamov type bound is presented in Section \ref{sec:TheoBound}.
Section \ref{sec:encoding} contains a new encoding scheme based on binary constant weight codes. The upper bound based on extremal set theory is presented in Section \ref{sec:UpperBound}.
We conclude in Section \ref{sec:conclusion}.

\section{Block permutation metric}\label{sec:pre}

In this section, we give some definitions and notations for permutation codes under block permutation metric.

Let $[n]$ denote $\{1,2,3,\ldots,n\}$. $\pi=(\pi(1),\pi(2),\ldots,\pi(n))$ is a permutation over $[n]$, known as the vector notation of a permutation.
The symbol $\circ$ denotes the composition of permutations. Specifically, for two permutations $\sigma$ and $\tau$, their composition, denoted by
$\sigma \circ \tau$, is the permutation with $\sigma \circ \tau(i)=\sigma(\tau(i))$ for all $i \in [n]$. All the permutations under this operation form
the noncommutative group $S_n$ known as the symmetric group on $[n]$ of size $\left|S_n\right|=n!$. The subsequence of $\sigma$ from indices $i$ to
$j$ is written as $\sigma \left[i;j\right] \triangleq \left(\sigma(i),\sigma(i+1),\ldots,\sigma(j)\right)$.

\begin{definition}
A permutation $\pi \in S_n$ is called {\it minimal} if and only if no consecutive elements in $\pi$ are also consecutive in the identity permutation
$e=(1,2,\ldots,n)$, i.e., for all $1 \leqslant i \leqslant n-1, \pi(i+1)\neq \pi(i)+1$. Denote the set of all the minimal permutations in $S_n$ as
$\mathcal{D}_{n}$.
\end{definition}

\begin{definition}\label{def:BPM}
The block permutation distance $d_{B}(\pi_{1},\pi_{2})$ between two permutations $\pi_{1},\pi_{2} \in S_n$ is equal to $d$ if
\begin{equation*}
  \pi_{1}=\left(\psi_{1},\psi_{2},\ldots,\psi_{d+1}\right), \pi_{2}=\left(\psi_{\sigma(1)},\psi_{\sigma(2)},\ldots,\psi_{\sigma(d+1)}\right),
\end{equation*}
where $\sigma\in\mathcal{D}_{d+1}$, $\psi_{k}=\pi_1\left[i_{k-1}+1:i_{k}\right]$ for $0=i_{0}<i_{1}<\cdots<i_{d}<i_{d+1}=n$ and $1\leqslant k\leqslant
d+1.$
\end{definition}

The definition suggests that in order to turn $\pi_{1}$ into $\pi_{2}$, one way is to first divide $\pi_{1}$ into $d+1$ segments
$\pi_{1}=\left(\psi_{1},\psi_{2},\ldots,\psi_{d+1}\right)$ and then perform a block level permutation of these segments according to a permutation
$\sigma\in\mathcal{D}_{d+1}$. The constraint of $\sigma$ being minimal indicates that $d_{B}(\pi_{1},\pi_{2})=d$ if and only if $d+1$ is the minimum
number of segments that $\pi_{1}$ needs to be divided into for such an operation. This definition is somehow not intuitive enough and thus Yang et al.
\cite{Yang2018Theoretical} found another way to characterize the block permutation distance explicitly by the {\it characteristic set} of a permutation.

\begin{definition}\label{def:char}
The characteristic set $A\left(\pi\right)$ for any $\pi \in S_n$ is defined as set of all the consecutive pairs in $\pi$, i.e.,
\begin{equation*}
A\left(\pi\right) \triangleq \left\{\left(\pi\left(i\right),\pi\left(i+1\right)\right)\mid 1\leqslant i < n\right\}.
\end{equation*}
\end{definition}

Note that the characteristic set of a permutation is equivalent to representing a permutation by a directed Hamiltonian path on $n$ vertices. That is,
the Hamiltonian path corresponding to $\pi$ is the set of edges in $\{(x,y)|x,y\in[n],(x,y)\in A(\pi)\}$. The following idea will be frequently used throughout the paper.
Given a subset of $A(\pi)$, the directed edges corresponding to the subset constitute a disjoint union of several directed paths
(an isolated vertex $v$ will be also regarded as a path starting and ending with $v$). Then $\pi$ should be obtained by
concatenating these directed paths into a directed Hamiltonian path.

Let $\mathcal{P}_{n}$ be the set $\{(i,j)|i\neq j, i\in [n], j\in[n]\}$. $\left|\mathcal{P}_{n}\right|=n(n-1)$. For each permutation $\pi \in S_n$, the
corresponding characteristic set $A\left(\pi\right)$ is then a subset of $\mathcal{P}_{n}$ of cardinality $\left|A\left(\pi\right)\right|=n-1$. The
block permutation metric can be characterized by the characteristic set and then some basic properties of the metric can be derived.
These are summarized in the following two lemmas proposed in \cite{Yang2018Theoretical}.

\begin{lemma}\label{lem:db}
For all $\pi_{1},\pi_{2}\in S_n$,
\begin{equation*}
  d_{B}\left(\pi_{1},\pi_{2}\right)=\left|A\left(\pi_{1}\right) \setminus A\left(\pi_{2}\right)\right|.
\end{equation*}
\end{lemma}

\begin{lemma}\label{lem:pro} For all $\pi_{1},\pi_{2},\pi_{3}\in S_n$, the block permutation distance $d_{B}$ satisfies the following properties:
\begin{enumerate}
  \item (Symmetry) $d_{B}\left(\pi_{1},\pi_{2}\right)=d_{B}\left(\pi_{2},\pi_{1}\right)$.
  \item (Left-invariance) $d_{B}\left(\pi_{3}\circ \pi_{1},\pi_{3} \circ \pi_{2}\right)=d_{B}\left(\pi_{1},\pi_{2}\right)$.
  \item (Triangle Inequality) $d_{B}\left(\pi_{1},\pi_{3}\right)\leqslant d_{B}\left(\pi_{1},\pi_{2}\right)+d_{B}\left(\pi_{2},\pi_{3}\right).$
\end{enumerate}
\end{lemma}

The following example shows how to compute the block permutation distance between two permutations following the terminologies above.

\begin{example}
  Let $\pi_{1} = (4,8,3,2,6,7,5,1,9)$, $\pi_{2}=(6,7,8,3,2,5,1,9,4).$ Their characteristic sets are
  \begin{flalign}
    A(\pi_{1})&=\{(4,8),(8,3),(3,2),(2,6),(6,7),(7,5),(5,1),(1,9)\}, \nonumber\\
    A(\pi_{2})&=\{(6,7),(7,8),(8,3),(3,2),(2,5),(5,1),(1,9),(9,4)\}, \nonumber
  \end{flalign}
  and thus we have $$d_{B}(\pi_{1},\pi_{2})=\left|A\left(\pi_{1}\right) \setminus A\left(\pi_{2}\right)\right|=|\{(4,8),(2,6),(7,5)\}|=3.$$

On the other hand, to compute $d_{B}(\pi_{1},\pi_{2})$ by Definition \ref{def:BPM}, we should find $\psi_{i}, 1\leqslant i \leqslant 4$
and $\sigma\in\mathcal{D}_{4}$ as follows:
$$\psi_{1}=(4), \psi_{2} =(8,3,2), \psi_{3}=(6,7), \psi_{4}=(5,1,9), \sigma=(3,2,4,1).$$
Then we have
\begin{flalign}
  \pi_{1} & =(\psi_{1},\psi_{2},\psi_{3},\psi_{4}), \nonumber \\
  \pi_{2} & =(\psi_{\sigma(1)},\psi_{\sigma(2)},\psi_{\sigma(3)},\psi_{\sigma(4)}), \nonumber
\end{flalign}
and thus $d_{B}(\pi_{1},\pi_{2})=3.$\\
\end{example}
Note that it is usually not easy to find such $\psi_{i}$ and $\sigma$ to compute the block permutation distance between two permutations, while finding the
difference between two characteristic sets is relatively easier. Next we introduce the permutation code under block permutation metric.

\begin{definition}
Given positive integers $n$ and $d$, $\mathcal{C}\subseteq S_n$ is called an $(n,d)$-permutation code under block permutation metric, if
$d_B(\sigma,\pi)\geqslant d$ for any two distinct permutations $\sigma,\pi\in \mathcal{C}$. Let $\mathcal{C}_{B}(n,d)$ denote the maximum size of an $(n,d)$-permutation code
$\mathcal{C}$.
\end{definition}

The best known upper bound and lower bound of $\mathcal{C}_{B}(n,d)$ are proposed in \cite{Yang2018Theoretical}, which are the so-called sphere-packing type bound and Gilbert-Varshamov
type bound. Both bounds are derived from the estimation on the size of a block permutation ball.

\begin{definition}\label{def:ball}
For given integers $n$, $t$ and a given center point $\pi \in S_n$, the $t$-block permutation ball centered at $\pi$ is defined as the set of all permutations $\sigma \in
S_n$, $d_{B}\left(\pi,\sigma \right)\leqslant t$. We denote the $t$-block permutation ball centered at $\pi$ as $b_{B}\left(n,t,\pi\right).$
\end{definition}

Note that by the left-invariance property of $d_{B}$, the size of $b_{B}\left(n,t,\pi\right)$ is independent of the center $\pi$ and thus we can
denote the size of the ball as $|b_{B}\left(n,t\right)|$.

\begin{lemma}\label{lem:ballsize}\textup{\cite{Yang2018Theoretical}}
For given integers $n$ and $t$, $t\leqslant n-\sqrt{n}-1$, denote the size of a $t$-block
permutation ball as $\left |b_{B}\left(n,t\right)\right |$, then we have
\begin{equation*}
  \prod\limits_{i=1}^{t}\left(n-i\right)\leqslant \left |b_{B}\left(n,t\right)\right| \leqslant \prod\limits_{i=0}^{t}\left(n-i\right).
\end{equation*}
\end{lemma}

\begin{lemma}\label{lem:GVSP}\textup{\cite{Yang2018Theoretical}}
For given integers $n$ and $t$, let $d=2t+1$, then we can bound $\mathcal{C}_{B}(n,d)$ as
\begin{equation*}
\frac{n!}{\left |b_{B}\left(n,2t\right)\right|} \leqslant \mathcal{C}_{B}(n,d) \leqslant \frac{n!}{\left |b_{B}\left(n,t\right)\right|}.
\end{equation*}
\end{lemma}

In \cite{Yang2018Theoretical} several constructions of $(n,d)$-permutation codes with $d=2t+1$ were presented, including a code of size
$\frac{n!}{q^{2d-3}}$, where $n(n-1) \leqslant q \leqslant 2n(n-1)$ is a prime number. Moreover \cite{Yang2018Theoretical} contains some explicit systematic constructions and decoding algorithms.

\section{Graph models}\label{sec:graph}

We use the standard terminologies and notations in graph theory. A graph $G$ consists of a set of vertices $V(G)$ and a set of edges $E(G)$.
Each edge is a pair of vertices. Two vertices $u$ and $v$ are called adjacent if there is an edge $\{u,v\}\in E\left(G\right)$.
We say that $H$ is a subgraph of $G$ if $V\left(H\right)\subset V(G)$ and $E(H)\subset E(G)$.
Furthermore if $H$ contains all edges of $G$ joining two vertices in $V(H)$, then $H$ is said
to be the subgraph of $G$ induced by $V(H)$. The neighborhood of a vertex $v$ is the set of all vertices adjacent to $v$, denoted by $\Gamma(v)$.
The neighborhood graph of $v$ is the subgraph induced by $\Gamma(v)$. The size of $|\Gamma(v)|$ is called the degree of the vertex $v$.
Let $\Delta(G)$ denote the maximum vertex degree. An independent set in a graph is a set of vertices where every
pair is nonadjacent. The size of the largest independent set in $G$ is called the independence number, denoted as $\alpha(G)$.

In this section we introduce a natural relationship between codes and independent sets of a corresponding graph. Take the set of all the codewords as the
vertex set of a graph. Two codewords with distance less than $d$ are connected via an edge. Then in an independent set of this graph, every two distinct
codewords have distance no less than $d$. Thus we have a correspondence between an independent set and a code with minimum distance $d$. The problem of
estimating the maximal size of a code turns into analyzing the independence number of the corresponding graph. This well-known approach has already been
shown to be powerful in studying several kinds of codes. Take the permutation code under Hamming metric as an example. Gao et al. \cite{2013GaoImprovementofGV}
improved the Gilbert-Varshamov bound by a factor of $\Omega(\log n)$, when the minimum distance $d$ is fixed and $n$ goes into infinity. Tail et al.
\cite{vardy} improved the Gilbert-Varshamov bound by a factor of $\Omega(n)$, when $\frac{d}{n}$ is fixed and $n$ goes into infinity. Recently, Wang et
al. \cite{2016WangPerms} used a coloring approach to analyze the independence number and improved the Gilbert-Varshamov bound by a factor of $\Omega(n)$
when the minimum distance $d$ is fixed and $n$ goes into infinity.

Here we introduce some results about the independence number of locally sparse graphs. A graph is called triangle-free if and only if the neighborhood
of every vertex is an independent set. 
Ajtai et al. \cite{1980AjtaiNoteRamsey} showed the relationship between triangle-free property and independence number in the following lemma.

\begin{lemma} Let $G$ be a graph with maximum degree $\Delta$. If $G$ is triangle-free, then we have
\begin{equation*}
\alpha(G)\geqslant \frac{|V(G)|}{8\Delta} \log_{2}\Delta.
\end{equation*}
\end{lemma}

\par In \cite{1985BollobasRandomgraphs} the lemma above was extended from triangle-free graphs into graphs with relatively few triangles.

\begin{lemma}
Let $G$ be a graph with maximum degree $\Delta$. If $G$ has at most $T$ triangles, then we have
\begin{equation*}
\alpha(G)\geqslant \frac {|V(G)|}{10\Delta}(\log_{2}\Delta- \frac{1}{2}\log_{2}(\frac{T}{|V(G)|})).
\end{equation*}
\end{lemma}

Note that a graph has relatively few triangles when the neighborhoods of its vertices are relatively sparse. Jiang and Vardy \cite{2004JiangGVbound}
generalized the results above for locally sparse graphs as follows.

\begin{lemma}\label{lem:JVsparse}
Let $G$ be a graph with maximum degree $\Delta$. Suppose for any vertex $v\in V(G)$, the subgraph induced by the neighborhood of $v$ has at most
$P$ edges, then we have
\begin{equation*}
\alpha(G)\geqslant \frac{|V(G)|}{10\Delta}(\log_{2}\Delta-\frac{1}{2}\log_{2}(\frac{P}{3})).
\end{equation*}
\end{lemma}

\section{An asymptotic improvement of the lower bound}\label{sec:TheoBound}

Before presenting the main results of this section, it should be noted that $\mathcal{C}_{B}(n,d)$ can be determined under some special cases.

\begin{theorem}
$\mathcal{C}_{B}(n,1)=n!$. $\mathcal{C}_{B}(n,2)=(n-1)!$. $\mathcal{C}_{B}(n,n-1)\leqslant n$ and equality holds if $n$ is not 3 or 5.
\end{theorem}

\begin{proof}
\begin{enumerate}
  \item Trivially take all the permutations in $S_n$ and we have $\mathcal{C}_{B}(n,1)=n!$.
  \item It is easy to check that for any two permutations $\pi$ and $\sigma$, $d_B(\pi,\sigma)=1$ if and only if $\sigma$ is a cyclic shift of $\pi$.
  That is, if $\pi=(\pi(1),\dots,\pi(n))$ and $d_B(\pi,\sigma)=1$, then $\sigma$ is of the form $\sigma=(\pi(t),\dots,\pi(n),\pi(1),\dots,\pi(t-1))$ for some $2\leqslant t \leqslant n$.
Under the operation of cyclic shifting, $S_n$ is divided into $(n-1)!$ equivalent classes where each class is known as a circular permutation.
By picking an arbitrary permutation from each equivalent class we obtain an $(n,2)$-permutation code of cardinality $(n-1)!$.
  \item For any two distinct permutations $\pi$ and $\sigma$ in an $(n,n-1)$-permutation code, their characteristic sets are disjoint according to Lemma \ref{lem:db}.
Since each characteristic set is a subset of $\mathcal{P}_{n}$ of cardinality $n-1$, $|\mathcal{P}_{n}|=n(n-1)$, then the number of codewords is at most $n$.
\begin{enumerate}
  \item Suppose $n$ is even, $n=2p$. Define $a_{2i-1}=2i-1$ for $1\leqslant i \leqslant p$ and $a_{2i}=2p-2i$ for $1\leqslant i \leqslant p-1$, i.e., $(a_{1},a_{2},\ldots,a_{n-1})=(1,2p-2,3,2p-4,\ldots,p,\ldots,4,2p-3,2,2p-1)$.
  For every $1\leqslant i \leqslant n$, let the $i$-th codeword be $(i,i+a_{1},i+a_{1}+a_{2},\dots,i+\sum\limits_{j=1}^{k}a_{j},\ldots,i+\sum\limits_{j=1}^{n-1} a_{j})$, where each entry is taken modulo $n$ (and note that we use `$n$' instead of `$0$' for some entry).
  It is routine to check that $\sum\limits_{j=1}^{2i}a_{j}\equiv -i \pmod{n}$, $1\leqslant i \leqslant p-1$ and $\sum\limits_{j=1}^{2i-1}a_{j} \equiv i \pmod{n}$, $1\leqslant i \leqslant p$. Therefore $\sum\limits_{j=1}^{k}a_{j}$ are distinct modulo $n$ for $1\leqslant k\leqslant n$. So these $n$ codewords defined above are indeed codewords in $S_{n}$.
  For every pair $(c,d)$ with $d-c\equiv a_{k} \pmod{n}$, it appears exactly once, in the $i$-th codeword with $c\equiv i+\sum\limits_{j=1}^{k-1} a_j$ and $d\equiv i+\sum\limits_{j=1}^{k} a_j$.
  \item Suppose $n$ is odd. To construct an $(n,n-1)$-permutation code of size $n$, consider the complete directed graph on $n+1$ vertices $[n]\cup\{\infty\}$.
  For each $\pi$, its characteristic set $A(\pi)$ also represents the directed Hamiltonian path on $n$ vertices. Further add the edges $(\infty,\pi(1))$ and $(\pi(n),\infty)$ into $A(\pi)$.
  Then each permutation corresponds to a directed Hamiltonian cycle on $[n]\cup\{\infty\}$. Thus an $(n,n-1)$-permutation code of size $n$ is equivalent
  to a Hamiltonian decomposition in the complete directed graph on $[n]\cup\{\infty\}$. Hamiltonian decomposition is a well studied topic, for example in \cite{1980TimothyDecomposition}.
  It has been shown that for odd integers $n \geqslant 7$, the edges of the complete directed graph on $n+1$ vertices can be partitioned into $n$ directed Hamiltonian circuits.
\end{enumerate}
  Therefore, $\mathcal{C}_{B}(n,n-1)\leqslant n$ and equality holds if $n$ is even or $n\geqslant 7$ is odd. Moreover, it can be easily checked that $\mathcal{C}_{B}(3,2)=2$ and $\mathcal{C}_{B}(5,4)=4$.
\end{enumerate}
\end{proof}

\begin{remark}
When $n+1$ is prime, there is another construction of an $(n,n-1)$-permutation code of size $n$ different from the one in the proof above. Consider the code $\{(i,2i,\dots,(n-1)i,ni):1\leqslant i \leqslant n\}$, with each entry modulo $(n+1)$. It is straightforward to check that every pair of $(a,b)$ appears exactly once (in the $i$th codeword, $i\equiv(b-a)\pmod{n+1}$).
\end{remark}

After solving these special cases, the rest of this section is devoted to improving the asymptotic lower bound of $\mathcal{C}_{B}(n,d)$ with $d\geqslant 3$ being a fixed constant, while $n$ approaches infinity.
The idea is to analyze the independence number of the corresponding block permutation graph, defined as follows.

\begin{definition}
For given positive integers $n$ and $d\geqslant 3$, the $(n,d)$-block permutation graph $\mathcal{G}_{n,d}$ is the graph with vertex set $S_n$ and edge set
$\{(\pi,\sigma):\pi\neq\sigma,d_B(\pi,\sigma)<d\}$.
\end{definition}

The codewords of an $(n,d)$-permutation code under block permutation metric are vertices of an independent set in $\mathcal{G}_{n,d}$. Conversely, any
independent set in $\mathcal{G}_{n,d}$ is an $(n,d)$-permutation code. To get a lower bound of $\mathcal{C}_{B}(n,d)$ via the graph theoretic approach using Lemma \ref{lem:JVsparse},
we need to calculate some parameters of the graph $\mathcal{G}_{n,d}$.

Let $\mathcal{H}_{n,d}$ be the subgraph induced by the neighborhood of the identity permutation $(1,2,3,\ldots,n)$, and let $R(n,k)$ be the set of all
permutations in $S_n$ which are exactly at distance $k$ from the identity, i.e.,
\begin{equation*}
 R(n,k)=\{\sigma\in S_n : d_{B}(\sigma,id)=k \}.
\end{equation*}
Then the induced subgraph $\mathcal{H}_{n,d}$ has the vertex set $V(\mathcal{H}_{n,d})=\bigcup\limits_{k=1}^{d-1} R(n,k)$. The size of $R(n,k)$ is a well-studied topic in \cite{2002MyersCounting}.

\begin{lemma}\textup{\cite{2002MyersCounting}}  For all integers $1\leqslant k\leqslant n-1$,
\begin{equation*}
  |R(n,k)|=k!\binom{n-1}{k}\sum_{i=0}^{k}(-1)^{k-i}\frac{(i+1)}{(k-i)!}.
\end{equation*}
\end{lemma}

Since $\binom{n}{a}=\Theta(n^{a})$ when $a$ is a fixed positive integer and $n$ goes to infinity, then asymptotically
$|R(n,k)|=\Theta(n^{k})$, $1\leqslant k \leqslant d-1$ and thus $|b_{B}(n,d-1)|=\sum\limits_{k=0}^{d-1}|R(n,k)|=\Theta(n^{d-1})$, when $d$ is fixed and $n$ goes to infinity.

To apply Lemma \ref{lem:JVsparse}, we already have $V(\mathcal{G}_{n,d})=n!$ and $\mathcal{G}_{n,d}$ is a regular graph of degree $\Delta=b_{B}\left(n,d-1\right)-1=\Theta(n^{d-1})$.
The remaining parameter to compute is $P(n,d)$, the number of edges in the induced subgraph $\mathcal{H}_{n,d}$.

\begin{lemma} \label{lem:2d-3}
For a fixed positive integer $d \geqslant 3$, $P(n,d)=O(n^{2d-3})$ when $n$ goes to infinity.
\end{lemma}

\begin{proof}
The number of vertices in $R(n,k)$ is asymptotically $\Theta(n^{k})$. Thus the number of edges connecting some $\pi\in R(n,k_1)$ and some $\sigma\in R(n,k_2)$ is $\Theta(n^{k_1+k_2})=O(n^{2d-3})$ as long as $k_1+k_2 \leqslant 2d-3$. Therefore, to prove the lemma we only need to focus on bounding the number of edges connecting some $\pi\in R(n,d-1)$ and some $\sigma\in R(n,d-1)$.

Consider the characteristic sets of such $\pi$ and $\sigma$. $|A(id)\setminus A(\pi)|=|A(id)\setminus A(\sigma)|=d-1$. Let $x(\pi,\sigma)$ be the number of consecutive pairs in $A(id)$ contained in neither $A(\pi)$ nor $A(\sigma)$, i.e.,
$$x(\pi,\sigma)=|\big(A(id)\setminus A(\pi)\big)\cap \big(A(id)\setminus A(\sigma)\big)|.$$

For a fixed $\pi\in R(n,d-1)$, the number of permutations $\sigma\in R(n,d-1)$ with $x(\pi,\sigma)=x$ is at most ${d-1 \choose x}{n-d\choose {d-1-x}}=\Theta(n^{d-1-x})$, since $A(id)\setminus A(\sigma)$ contains exactly $x$ pairs out of the $d-1$ pairs in $A(id)\setminus A(\pi)$ and $d-1-x$ pairs out of the $n-d$ pairs in $A(id)\cap A(\pi)$. Recall that $|R(n,d-1)|=\Theta(n^{d-1})$ and then the number of edges connecting $\pi,\sigma \in R(n,d-1)$ with $1\leqslant x(\pi,\sigma)\leqslant d-1$ is at most $\Theta(n^{2d-3})$. Therefore, to prove the lemma we only need to focus on bounding the number of edges connecting some $\pi\in R(n,d-1)$ and some $\sigma\in R(n,d-1)$, with $x(\pi,\sigma)=0$. Now we claim that in fact there are no such edges.

Since $x(\pi,\sigma)=0$, then $\big(A(id)\setminus A(\sigma)\big)\subset \big(A(\pi)\setminus A(\sigma)\big)$ and thus $d_B(\pi,\sigma)\geqslant d-1$. If $\pi$ and $\sigma$ are connected, then it must hold that $d_B(\pi,\sigma)=d-1$ and
$$A(id)\setminus A(\sigma)~=~A(\pi)\setminus A(\sigma)$$
and simultaneously
$$A(id)\setminus A(\pi)~=~A(\sigma)\setminus A(\pi).$$

Now consider the $n-d$ pairs in $A(\pi)\cap A(\sigma)$. In the graph with vertex $[n]$, label all the directed edges $(x,y)$ where $(x,y)\in A(\pi)\cap A(\sigma)$ and call this graph $\mathcal{G}$. The union of $A(\pi)\cap A(\sigma)$ and $A(id)\setminus A(\sigma)$ is $A(\pi)$, the directed Hamiltonian path corresponding to $\pi$. Therefore $\mathcal{G}$ is a union of $d$ non-intersecting directed paths (there may exist isolated vertices and each isolated vertex is also considered as a directed path), where the $j$th path is denoted as $P_j=(x_j\rightarrow\cdots\rightarrow y_j)$, indicating that it starts with $x_j$ and ends with $y_j$, $1\leqslant j \leqslant d$. The directed Hamiltonian path corresponding to $\pi$ is then a concatenation of these paths and without loss of generality it can be written as $P_1\rightarrow P_2\rightarrow \cdots \rightarrow P_d$. Since the edges connecting the $P_j$'s arise from $A(id)\setminus A(\sigma)$, then it implies that $x_{j+1}=y_j+1$ for $1\leqslant j \leqslant d-1$.

Now since the directed Hamiltonian path corresponding to $\sigma$ is also formed by using the $d-1$ edges in $A(id)\setminus A(\pi)$ to connect the $P_j$'s, then there are only two cases. The first case is when $x_1\neq y_d+1$, then there is only a unique way to connect the $P_j$'s via edges corresponding to consecutive pairs, i.e., $\sigma=\pi$. The other case is when $x_1=y_d+1$ and the directed Hamiltonian path corresponding to $\sigma$ will be of the form $P_t\rightarrow P_{t+1}\rightarrow \cdots \rightarrow P_d \rightarrow P_1 \rightarrow \cdots \rightarrow P_{t-1}$. However, since $d\geqslant 3$, then $\sigma$ and $\pi$ will share $d-2$ edges $\{(y_j,x_{j+1})|j\neq t-1, 1\leqslant j \leqslant d-1\}$, which contradicts to $x(\pi,\sigma)=0$.

Therefore, the last kind of edges we focus on do not exist at all and the total number of edges in the graph $\mathcal{H}_{n,d}$ is $P(n,d)=O(n^{2d-3})$.
\end{proof}

Now we are ready to apply Lemma \ref{lem:JVsparse} to obtain the new lower bound of $\mathcal{C}_{B}(n,d)$.

\begin{theorem}
When $d$ is fixed, $d \geqslant 3$ and $n$ goes into infinity, there exists an $(n,d)$-permutation code under block permutation metric with size
$$\mathcal{C}_{B}(n,d)=\alpha(\mathcal{G}_{n,d})\geqslant \frac{n!}{10\Delta}(\log_{2}\Delta-\frac{1}{2}\log_{2}(\frac{P(n,d)}{3})) = \Omega(\frac{n!\log{n}}{n^{d-1}}).$$
Particularly, it improves the Gilbert-Varshamov bound by a factor of $\Omega(\log(n)).$
\end{theorem}

\begin{proof}
  Using our graph notation, the Gilbert-Varshamov bound is
\begin{equation*}
A_{GV}(n,d):=\frac{n!}{1+\Delta(n,d)}=\Theta(\frac{n!}{n^{d-1}}).
\end{equation*}
By Lemma \ref{lem:JVsparse} and Lemma \ref{lem:2d-3}, we have
\begin{flalign}
\frac{\alpha(\mathcal{G}_{n,d})}{A_{GV}(n,d)}  & \geqslant
\frac{\frac{n!}{10\Delta(n,d)}(\log_{2}\Delta(n,d)-\frac{1}{2}\log_{2}(\frac{P(n,d)}{3}))}{\frac{n!}{1+\Delta(n,d)}} \nonumber \\
\geqslant & \frac{1}{10} \log_{2}(\frac{\Delta(n,d)}{\sqrt{\frac{P(n,d)}{3}}}) \geqslant \frac{1}{10} \log_{2}(\frac{c_{b}n^{d-1}}{c_{s}n^{d-\frac{3}{2}}})
=c\log(n).  \nonumber
\end{flalign}

Hence we have

$$\frac{\alpha(\mathcal{G}_{n,d})}{A_{GV}(n,d)}=\Omega(\log(n)). $$\\

where $c_{b}$, $c_{s}$ and $c$ are constants independent of $n$.
\end{proof}

\section{Construction} \label{sec:encoding}

In this section, we propose a new construction of permutation codes under block permutation metric. The main idea arises from constructing constant weight binary codes under Hamming metric.

Recall that $\mathcal{P}_n=\{(x,y):x\neq y, x,y\in[n]\}$ and $|\mathcal{P}_n|=n(n-1)$. Suppose $q\geqslant n(n-1)/2$ is a prime number.
From \emph{Bertrand's postulate}, there is always such a $q$, $n(n-1)/2\leqslant q \leqslant n(n-1)$.

Let $\mathcal{V}:\mathcal{P}\rightarrow \mathbb{F}_q$ be a map from $\mathcal{P}$ to the finite field $\mathbb{F}_q$ such that for distinct pairs $(x,y)$ and $(x',y')$, $\mathcal{V}(x,y)=\mathcal{V}(x',y')$ if and only if
$x'=y$ and $y'=x$. The range of $\mathcal{V}$ has size $n(n-1)/2$ and can be satisfied since we set $q\geqslant n(n-1)/2$.

Then for any permutation $\pi\in S_n$, $\mathcal{V}$ maps its characteristic set $A(\pi)=\{(\pi(i),\pi(i+1))\mid 1\leqslant i < n\}$ into $\{\mathcal{V}((\pi(i),\pi(i+1))\mid 1\leqslant i < n\}$, which is a subset of $\mathbb{F}_q$ of cardinality $n-1$. Denote these $n-1$ elements as $\gamma_{1},\gamma_{2},\dots,\gamma_{n-1}$.

We then define a map $F$ from $S_n$ to $\mathbb{F}_{q}^{d-1}$ as follows:
\begin{equation*}
  F(\pi)=(F_{1}(\pi),F_{2}(\pi),...,F_{d-1}(\pi)),
\end{equation*}
 where
\begin{flalign*}\label{map:F}
    F_{1}(\pi)&=\sum_{1\leqslant i \leqslant n-1}\gamma_{i},  \\
    F_{2}(\pi)&=\sum_{1\leqslant i < j \leqslant n-1}\gamma_{i}\gamma_{j},  \\
    F_{3}(\pi)&=\sum_{1\leqslant i < j < k \leqslant n-1}\gamma_{i}\gamma_{j}\gamma_{k}, \\
              & ...\nonumber
\end{flalign*}

\begin{theorem}
For any two distinct permutations $\pi,\sigma \in S_n$, if $F(\pi)=F(\sigma)$, then $d_B(\pi,\sigma)\geqslant d$.
\end{theorem}

\begin{proof}
Suppose on the contrary that there exist two distinct permutations $\pi,\sigma \in S_n$ such that $F(\pi)=F(\sigma)$ and $d_B(\pi,\sigma)=\delta<d$.
Recall that $d_B(\pi,\sigma)=|A(\pi)\setminus A(\sigma)|= |A(\sigma)\setminus A(\pi)|$.
Therefore $\mathcal{V}$ maps the set $A(\pi)\setminus A(\sigma)$ into a subset $\{\alpha_{1},\alpha_{2},\dots,\alpha_{\delta}\}$ and similarly
$\mathcal{V}$ maps the set $A(\sigma)\setminus A(\pi)$ into a subset $\{\beta_{1},\beta_{2},\dots,\beta_{\delta}\}$.

The condition $F(\pi)=F(\sigma)$ will infer the following equations.
 \begin{flalign*}
   \zeta_{1}   &=\sum_{1\leqslant i \leqslant \delta}\alpha_{i}=\sum_{1\leqslant i \leqslant \delta}\beta_{i},   \\
   \zeta_{2}   &=\sum_{1\leqslant i < j \leqslant \delta}\alpha_{i}\alpha_{j}=\sum_{1\leqslant i < j \leqslant \delta}\beta_{i}\beta_{j}, \\
               & \ldots                                               \\
   \zeta_{d-1} &=\sum_{i_{1}<\ldots<i_{d-1}}\alpha_{i_{1}}\ldots\alpha_{i_{d-1}}=\sum_{i_{1}<\ldots<i_{d-1}}\beta_{i_{1}}\ldots\beta_{i_{d-1}}.
 \end{flalign*}

Consider the polynomial $x^{\delta}-\zeta_{1}x^{\delta -1}+\zeta_{2}x^{\delta -2}-\cdots +(-1)^{\delta +1}\zeta_{\delta}=\prod_{1\leqslant i \leqslant \delta}(x-\alpha_i)=\prod_{1\leqslant i \leqslant \delta}(x-\beta_i)$.
Then $\{\alpha_{1},\alpha_{2},\dots,\alpha_{\delta}\}$ and $\{\beta_{1},\beta_{2},\dots,\beta_{\delta}\}$ are both the zeros of this polynomial and thus these two sets are identical.

Consider the complete directed graph with vertex set $[n]$ where each permutation corresponds to a directed Hamiltonian path indicated by its characteristic set. Now the path indicating $\pi$ and the path indicating $\sigma$ share $n-1-\delta$ directed edges in $A(\pi)\cap A(\sigma)$. Due to the property of the map $\mathcal{V}$, the set $\mathcal{E}$ of edges (without considering directions at this moment) corresponding to the pairs $\{\alpha_{1},\alpha_{2},\dots,\alpha_{\delta}\}=\{\beta_{1},\beta_{2},\dots,\beta_{\delta}\}$ are uniquely determined. With the given directions on the edges $A(\pi)\cap A(\sigma)$, there is a unique way to choose the directions for the edges in $\mathcal{E}$ to obtain a Hamiltonian path. Therefore $\pi$ should be the same as $\sigma$, a contradiction.
\end{proof}

Therefore, we can construct $(n,d)$-permutation codes under block permutation metric as follows.
\begin{theorem}\label{thm:constructcode}
  For every $\mathbf{f}\in\mathbb{F}_{q}^{d-1}$, $C_{\mathbf{f}}(n,d)=\{\pi|\pi\in S_{n}, F(\pi)=\mathbf{f}\}$ is an $(n,d)$-permutation code under block permutation metric.
\end{theorem}

Consider all the vectors $\mathbf{f}\in\mathbb{F}_{q}^{d-1}$ and then $\{C_{\mathbf{f}}(n,d): \textbf{f}\in \mathbb{F}_{q}^{d-1}\}$ is a partition of $S_{n}$, where each
component $C_{\emph{\textbf{f}}}(n,d)$ is a permutation code under block permutation metric. Suppose $C_{\mathbf{f}_{\max}}(n,d)$ is the one with
maximal size, then by pigeonhole principle, we obtain that $|C_{\mathbf{f}_{\max}}(n,d)| \geqslant \frac{n!}{|\mathbb{F}_{q}^{d-1}|} =
\frac{n!}{q^{d-1}} = \frac{n!}{n^{2d-2}}.$

In \cite{Yang2018Theoretical}, Yang et al. constructed a permutation code of size $\frac{n!}{q^{2d-3}}=\frac{n!}{n^{4d-6}}$, where $q$ is a prime number
such that $n(n-1) \leqslant q \leqslant 2n(n-1)$. So our construction improves the size of permutation codes by a factor of $\Theta(n^{2d-4})$.

\section{An upper bound} \label{sec:UpperBound}

In this section, we obtain a new upper bound by means of analyzing the characteristic sets of the codewords.
Recall that for each permutation $\pi\in S_n$, its characteristic set $A(\pi)=\{ (\pi(i),\pi(i+1))|1\leqslant i <n \}$ is a subset of
$\mathcal{P}_{n}$ of cardinality $|A(\pi)|=n-1$.
Denote $I(\pi_{1},\pi_{2})=|A(\pi_{1})\cap A(\pi_{2})|$, then we have
\begin{lemma}\label{lem:intersection}
For any $\pi_{1},\pi_{2}\in S_n$, $d_{B}(\pi_{1},\pi_{2})\geqslant d$ if and only if $I(\pi_{1},\pi_{2}) \leqslant n-d-1$.
\end{lemma}

Given an $(n,d)$-permutation code $\mathcal{C}$, let $\mathcal{F}$ be the collection of all the characteristic sets $A(\pi)$ of the codewords, i.e.,
$\mathcal{F}=\{ A(\pi)|\pi\in\mathcal{C} \}$. We translate the problem of analyzing the bound of codes into the following extremal set theory problem:
find the maximal size of a family $\mathcal{F}$ of $(n-1)$-subsets of $\mathcal{P}_{n}$ satisfying that the intersection of each pair of subsets is at most $n-d-1$.
Then we can obtain an upper bound of a new type as follows.

\begin{theorem}\label{thm:upperbound}
    For given integers $n$ and $d$,
    \begin{equation*}
      |\mathcal{F}|\leqslant \frac{\binom{n}{d}\binom{n}{d}(n-d)!}{\binom{n-1}{n-d}}.
    \end{equation*}
\end{theorem}

\begin{proof}
  Let $T(n,d)$ be the family of all possible $(n-d)$-subsets of some $A(\pi)$, $\pi\in S_n$. Each $A(\pi)\in\mathcal{F}$ contains $\binom{n-1}{n-d}$
  such subsets. By Lemma \ref{lem:intersection}, any $(n-d)$-subset in $T(n,d)$ is contained in the characteristic set of at most one codeword.
  Therefore $|\mathcal{F}|\binom{n-1}{n-d}\leqslant |T(n,d)|$.

  The remaining problem is to estimate $|T(n,d)|$. For each set $A\in T(n,d)$, consider the $n\times n$ matrix $M=(m_{i,j})$ where
\begin{equation*}
    m_{i,j}=
  \begin{cases}
    1, & \mbox{if pair } (i,j)\in A,  \\
    0, & \mbox{otherwise}.
  \end{cases}
\end{equation*}
Since $A$ is an $(n-d)$-subset of some $A(\pi)$, $\pi\in S_n$, then the matrix should contain exactly $n-d$ entries of `1' and the weight of each column and row is at most 1.
Then the number of distinct $A$ is upper bounded by the number of ways to select $n-d$ rows and $n-d$ columns and construct a permutation matrix from the chosen sub-matrix.
Hence $T(n,d)\leqslant \binom{n}{n-d}\binom{n}{n-d}(n-d)!$. Therefore we have
$|\mathcal{F}|\leqslant\frac{\binom{n}{d}\binom{n}{d}(n-d)!}{\binom{n-1}{n-d}}$.
\end{proof}

By Lemma \ref{lem:ballsize} and Lemma \ref{lem:GVSP}, if $t \leqslant n-\sqrt{n}-1$, denote the sphere-packing bound as $A_{SP}(n,2t+1)$, which falls in the
range
$$\frac{n!}{\prod\limits_{i=0}^{t}(n-i)}\leqslant A_{SP}(n,2t+1) \leqslant \frac{n!}{\prod\limits_{i=1}^{t}(n-i)}.$$

Denote our new type upper bound $\frac{\binom{n}{d}\binom{n}{d}(n-d)!}{\binom{n-1}{n-d}}$ as $A_{new}(n,d)$.

\begin{corollary}
Given $n$ and $d=2t+1$, if $t\leqslant n-\sqrt{n}-1$, $n\cdot\prod\limits_{i=0}^{t}\left(n-i\right)\leqslant d\cdot d!$ and $d\leqslant n-1$, then
$A_{new}(n,d)\leqslant A_{SP}(n,d)$.
\end{corollary}

In Table \ref{table:UP} we list several cases for small parameters as supporting evidences to show that the new bound in Theorem \ref{thm:upperbound} works better than
sphere-packing bound when $d$ is relatively close to $n$. Note that the values of sphere-packing bound in this table say that the size of codes is upper bounded by some value $x$, where $x$ is not less than the values shown in the table. (For example, the size of a (13,9)-code is upper bounded by $x$, where $x\geqslant 40320$. It doesn't necessarily suggest that the size of a (13,9)-code is upper bounded by $40320$. Our new result indicates that the size of a (13,9)-code is upper bounded by 24787, which is indeed an improvement over the sphere-packing bound.)

\begin{table}[!h]\label{table:bound}
  \centering
  \caption{A comparison of new bound and sphere-packing bound with some small parameters} \label{table:UP}

  \begin{tabular}{cccc|cccc}
  \toprule
  n & d & Sphere-packing bound & Theorem \ref{thm:upperbound} &n & d & Sphere-packing bound & Theorem \ref{thm:upperbound}\\
  \midrule
   13 & 9 & $\geqslant$ 40320 & {\bfseries 24787} &18 & 11 & $\geqslant$ 479001600 & {\bfseries 262461363} \\
   15 & 11 & $\geqslant$ 362880 & {\bfseries 44672} &18 & 13 & $\geqslant$ 39916800 & {\bfseries 1423607}\\
   16 & 11 & $\geqslant$ 3628800 & {\bfseries 762415} & 19 & 11 & $\geqslant 6227020800$  & {\bfseries 5263805324}\\
   17 & 11 & $\geqslant$ 39916800 & {\bfseries 13771113} &19 & 13 & $\geqslant$ 479001600 & {\bfseries 28551213}\\
   17 & 13 & $\geqslant$ 3628800 & {\bfseries 74696} & 20 & 13 & $\geqslant 6227020800$  & {\bfseries 601078154}\\

   \bottomrule
  \end{tabular}
\end{table}

\section{Conclusion}\label{sec:conclusion}

In this paper, we establish the correspondence between permutation codes and the independent sets of block permutation graphs. Using this approach,
we improve the Gilbert-Varshamov bound asymptotically by a factor of $\Omega(\log n)$ when the minimum distance $d$ is fixed and $n$ goes into infinity.
As for the upper bound, we clarify the relationship between block permutation distance of permutations and the intersection of their characteristic
sets. Using some counting methods, we derive an upper bound of a new type, which beats the sphere-packing bound when $d$ is relatively close to $n$.
Moreover, we present the existence of a permutation code which improves the size of the known result by a factor of $\Theta(n^{2d-4})$.
Explicit encoding schemes achieving this size are considered for future research.

\bibliographystyle{plain}
\bibliography{ivanref}
\end{document}